\documentclass[11pt,a4paper,reqno]{amsart}%
\usepackage{amsthm,amsmath,amsfonts,amssymb,xcolor,amsxtra,appendix,bookmark,dsfont,bm,mathrsfs}
\usepackage[latin1]{inputenc}
\usepackage{color} 
\usepackage{amssymb}\usepackage{graphicx}
\usepackage{tikz}
\usepackage{hyperref}
\theoremstyle{plain}
\newtheorem{theorem}{Theorem}
\newtheorem{lemma}[theorem]{Lemma}

\newtheorem{proposition}[theorem]{Proposition}

\newtheorem{definition}{Definition}

\theoremstyle{remark}
\newtheorem{remark}[theorem]{Remark}

\allowdisplaybreaks

\title[Global time-analytic strong solutions for a class of 3-wave kinetic equations]{Global time-analytic strong solutions for a  class of 3-wave kinetic equations}

\author[N. G. Hien]{Nguyen Gia Hien}
\address{Department of Mathematics, Texas A\&M University, College Station, TX 77843, USA}
\email{giahien-nguyen@tamu.edu} 

\author[G. Staffilani]{Gigliola Staffilani}
\address{Department of Mathematics, Massachusetts Institute of Technology, Cambridge, MA 02139, USA}
\email{gigliola@math.mit.edu} 
\thanks{G.S. is  funded in part by  the NSF grants DMS-2052651, DMS-2306378 and the Simons Foundation through the Simons Collaboration on Wave Turbulence.}

\author[M.-B. Tran]{Minh-Binh Tran}
\address{Department of Mathematics, Texas A\&M University, College Station, TX 77843, USA}
\email{minhbinh@tamu.edu} 
\thanks{M.-B. T is  funded in part by  a   Humboldt Fellowship,   NSF CAREER  DMS-2303146, and NSF Grants DMS-2204795, DMS-2305523,  DMS-2306379.}

\begin{document}
	\date{\today}
	\begin{abstract} 
		We study a class of 3-wave kinetic equations arising in wave turbulence theory, with  regularized kernels. For radial, nonnegative initial data, we construct an exact global-in-time strong solution which remains nonnegative and is analytic with respect to time. The proof combines a careful analysis of the resonant interaction surfaces with a time power-series construction and a continuation argument based on the conservation of the energy moment.

	\end{abstract}
	\maketitle
	
	\section{Introduction}\label{intro}

	Wave turbulence theory provides a statistical description of the transfer of
	energy across scales in weakly nonlinear dispersive wave systems. Since its
	early development, it has become a central framework in mathematical physics,
	with applications to rotating fluids, magnetohydrodynamics, Alfv\'en waves in
	the solar wind, plasma turbulence, and wave processes arising in fusion-related
	models. The foundations of the theory can be traced back to the pioneering work
	of Peierls \cite{Peierls:1993:BRK}, and were further developed through the
	seminal contributions of Benney and Saffman \cite{benney1966nonlinear},
	Zakharov and Falkovich \cite{zakharov1967weak}, Benney and Newell
	\cite{benney1969random}, and Hasselmann
	\cite{hasselmann1962non,hasselmann1974spectral}. These developments led to the
	formulation of wave kinetic equations, in particular the classical 3-wave
	and 4-wave kinetic equations, which describe the resonant exchange and
	redistribution of energy among weakly interacting modes. In recent years, the
	rigorous derivation of such kinetic equations from underlying dispersive wave
	systems has seen major progress, most notably in a series of works by Deng and
	Hani
	\cite{deng2019derivation,deng2021propagation,deng2023long,deng2023,
		deng2021full}. For broader accounts of the physical theory, its scaling laws,
	and its applications, we refer to
	\cite{Nazarenko:2011:WT,PomeauBinh,zakharov2012kolmogorov}.

	In this work, we consider the following 3-wave kinetic equation
	\begin{equation}\label{3wave}
		\begin{aligned}
			\partial_t f(t,k) \ = & \ Q_{3w}[f](k), \ \ \ (t,k)\in (0,\infty)\times\mathbb{R}^d, \\\
			f(0,k) \ = & \ f_0(k),  k\in\mathbb{R}^d, \ \  d\ge 3
		\end{aligned}
	\end{equation}
	in which the collision operator is of the form 
	\begin{equation}\label{def-Qf}
		Q_{3w}[f](k) \ := \ \iint_{\mathbb{R}^{2d}} \Big[ R_{k,k_1,k_2}[f] - R_{k_1,k,k_2}[f] - R_{k_2,k,k_1}[f] \Big]\mathrm{d}k_1\mathrm{d}k_2
	\end{equation}
	with
	\begin{equation}\label{def-R}
		\begin{aligned}
			R_{k,k_1,k_2} [f] \ := \ V_{k,k_1,k_2}\delta(k-k_1-k_2)\delta(\omega_k -\omega_{k_1}-\omega_{k_2})(f_1f_2-ff_1-ff_2)
		\end{aligned}
	\end{equation}
	with the short-hand notation $f = f(t,k)$ and $f_j = f(t,k_j)$. The Dirac delta function $\delta(\cdot)$ is to ensure the following resonant conditions for the wavenumbers:
	\begin{equation}\label{cv} 
		k = k_1 + k_2 , \qquad \omega_k = \omega_{k_1} + \omega_{k_2},
	\end{equation}
	with $\omega_k$ denoting the dispersion relation of the waves.

	3-wave kinetic equations appear in a broad range of physical models and
	have been investigated from several analytical perspectives. They arise, for
	instance, in the study of phonon interactions in anharmonic crystal lattices
	\cite{CraciunBinh,EscobedoBinh,GambaSmithBinh,tran2020reaction}, capillary
	wave turbulence \cite{nguyen2017quantum}, and stratified flows in geophysical
	fluid dynamics \cite{GambaSmithBinh,kim2025wave}. Closely related 3-wave
	kinetic structures also occur in kinetic models connected with Bose--Einstein
	condensation
	\cite{AlonsoGambaBinh,cortes2020system,EPV,escobedo2023linearized1,
		escobedo2023linearized,ToanBinh,nguyen2017quantum,Binh1,staffilani2025formation}. Entropy structures, energy cascade and ergodicity of those equations have been studied in \cite{rumpf2021wave,soffer2019energy,staffilani2026entropy}.  Alongside these
	analytical developments, there has been growing numerical interest in
	3-wave kinetic dynamics, including recent computational and data-driven
	studies
	\cite{banks2025new,das2024numerical,tran2026analysis,walton2023numerical,walton2024numerical,walton2022deep}.

	We consider  the following power-law dispersion relation
	\begin{equation}\label{def-Epp}
		\omega_k = \omega(k) = |k|^\gamma, \qquad 1 < \gamma< 2.
	\end{equation}
	
	\noindent The precise form of the collision kernel $V_{k,k_1,k_2}$ is 
	\begin{equation}\label{def-V}
		V_{k,k_1,k_2} := \mathscr M(k)\,\mathscr M(k_1)\,\mathscr M(k_2),
	\end{equation}
	where
	\begin{equation}\label{M}
		\mathscr M(k) := |k|^{\gamma+d-2}\, \exp(-\mathscr A|k|),
	\end{equation}
	for a constant number $\mathscr A>0$.
	
	\begin{remark}[On the choice of the collision kernel]
		The kernel \eqref{def-V}--\eqref{M} should be understood as a model
		interaction kernel. The present choice of \(\mathscr M(k)\) 
		is made in order to retain the resonant 3-wave collision structure while
		providing a nonnegative, radial, separable, and high-frequency localized
		kernel. The exponential factor gives strong decay at large wave numbers,
		whereas the power-law prefactor controls the behavior near the origin and is
		compatible with the weighted estimates used below. Thus the results of this
		paper should be interpreted as applying to a regularized model class of
		3-wave kinetic equations.
	\end{remark}
	The main contribution of this paper is the construction of a global-in-time,
	nonnegative solution to the 3-wave kinetic equation \eqref{3wave} which is
	analytic with respect to time.    The theorem is proved for a regularized radial model, whose kernel retains the essential resonant 3-wave structure while allowing uniform control of the surface collision integrals.

	\begin{definition}[Finite-energy strong radial solution]
		Let \(T\in(0,\infty]\). A function
		\[
		f:[0,T)\times \mathbb R^d\to \mathbb R
		\]
		is called a finite-energy strong radial solution of the 3-wave kinetic equation
		\[
		\partial_t f=Q_{3w}[f],
		\qquad
		f(0,\cdot)=f_0,
		\]
		on \([0,T)\) if the following conditions hold.
		
		\begin{enumerate}
			\item For every \(t\in[0,T)\), the function
			\(f(t,\cdot)\) is radial. That is, there exists a measurable function
			\[
			F(t,\cdot):[0,\infty)\to \mathbb R
			\]
			such that
			\[
			f(t,k)=F(t,|k|)
			\]
			for a.e. \(k\in\mathbb R^d\).
			
			\item  One has
			\[
			f\in C\bigl([0,T);L^\infty_{\mathrm{rad}}(\mathbb R^d)\bigr),
			\]
			and for every \(0<T_0<T\),
			\[
			\sup_{0\le t\le T_0}
			\int_{\mathbb R^d} |f(t,k)|\,\omega(k)\,\mathrm dk
			<\infty .
			\]
			Equivalently, in radial variables,
			\[
			\sup_{0\le t\le T_0}
			\int_0^\infty |F(t,r)|\,r^{\gamma+d-1}\,\mathrm dr
			<\infty .
			\]
			
			\item  For every \(t\in[0,T)\),
			the surface integrals defining \(Q_{3w}[f(t)]\) are absolutely convergent
			for a.e. \(k\in\mathbb R^d\), and
			\[
			Q_{3w}[f(t)]\in L^\infty_{\mathrm{rad}}(\mathbb R^d).
			\]
			Moreover,
			\(
			t\mapsto Q_{3w}[f(t)]
			\)
			belongs to
			\(
			C\bigl([0,T);L^\infty_{\mathrm{rad}}(\mathbb R^d)\bigr).
			\)
			
			\item The map
			\(
			t\mapsto f(t,\cdot)
			\)
			belongs to
			\(
			C^1\bigl([0,T);L^\infty_{\mathrm{rad}}(\mathbb R^d)\bigr),
			\)
			and satisfies
			\[
			\frac{d}{dt}f(t,\cdot)=Q_{3w}[f(t)]
			\]
			in \(L^\infty_{\mathrm{rad}}(\mathbb R^d)\) for every \(t\in[0,T)\).
			
			\item One has
			\[
			f(0,\cdot)=f_0
			\]
			in \(L^\infty_{\mathrm{rad}}(\mathbb R^d)\).
		\end{enumerate}
		
		If, in addition,
		\[
		f(t,k)\ge 0
		\]
		for a.e. \(k\in\mathbb R^d\) and every \(t\in[0,T)\), then \(f\) is called a
		nonnegative  strong radial solution.
	\end{definition}
	
	\begin{theorem}\label{maintheorem}
		Assume that \(d\geq 3\), \(1<\gamma< 2\), and let
		\[
		\omega(k)=|k|^\gamma,
		\qquad
		\mathscr M(k)=|k|^{\gamma+d-2}\exp(-\mathscr A|k|),
		\qquad \mathscr A>0.
		\]
		Let the initial datum \(f_0\) be radial and nonnegative, and assume that
		\[
		M:=\|f_0\|_{L^\infty(\mathbb R^d)}<\infty
		\]
		and
		\[
		\mathcal M_1
		:=
		\int_0^\infty f_0(r)\,r^{\gamma+d-1}\,\mathrm dr
		<\infty .
		\]
		Then the 3-wave kinetic equation \eqref{3wave}, with collision
		operator defined by \eqref{def-Qf}, admits an exact unique global strong solution
		\(f(t,k)\) such that
		\[
		f(t,k)\geq 0
		\]
		for a.e. \(k\in\mathbb R^d\).
		
		Moreover,
		\(
		t\mapsto f(t,\cdot)
		\)
		is real analytic as an \(L^\infty_{\mathrm{rad}}(\mathbb R^d)\)-valued map on
		\((0,\infty)\), and is right-analytic at \(t=0\) in
		\(L^\infty_{\mathrm{rad}}(\mathbb R^d)\). More precisely, for every
		\(t_0>0\), there exist \(\tau_{t_0}>0\) and coefficients
		\(
		\mathcal A_n^{(t_0)}\in L^\infty_{\mathrm{rad}}(\mathbb R^d),
		n=0,1,2,\ldots,
		\)
		such that
		\[
		f(t_0+s,\cdot)
		=
		\sum_{n=0}^{\infty}
		s^n\mathcal A_n^{(t_0)}
		\quad
		\text{in }L^\infty_{\mathrm{rad}}(\mathbb R^d)
		\]
		for every \(|s|<\tau_{t_0}\). At \(t_0=0\), there exist
		\(\tau_0>0\) and coefficients
		\(
		\mathcal A_n^{(0)}\in L^\infty_{\mathrm{rad}}(\mathbb R^d)
		\)
		such that
		\[
		f(s,\cdot)
		=
		\sum_{n=0}^{\infty}
		s^n\mathcal A_n^{(0)}
		\quad
		\text{in }L^\infty_{\mathrm{rad}}(\mathbb R^d)
		\]
		for every \(0\leq s<\tau_0\).
	\end{theorem}

	
	Let us now turn to the collision operator \eqref{def-Qf}. The collision operator involves surface integrals. Precisely, we introduce functions
	\begin{equation}\label{FunctionH}
		\begin{aligned}
			\mathscr H_0^k(v):=\omega_{k-v}+\omega_v-\omega_k,
			\qquad
			\mathscr H_1^k(v):=\omega_k+\omega_v-\omega_{k+v}.
		\end{aligned}
	\end{equation}
	and the energy surfaces, dictated by the resonant conditions \eqref{cv}, 
	\begin{equation}\label{def-Sp}
		\begin{aligned}
			\mathscr S_k
			&:=
			\Big\{
			v\in \mathbb R^d:
			\mathscr H_0^k(v)=0
			\Big\}, \ \ \ \
			\mathscr S'_k
			:=
			\Big\{
			v\in \mathbb R^d:
			\mathscr H_1^k(v)=0
			\Big\}.
		\end{aligned}
	\end{equation}
	The collision operator $Q_{3w}[f]$ then reduces to 
	\begin{equation}\label{WT-Q} 
		Q_{3w}[ f](k)= \int_{\mathscr S_k} R_{k,k - k_2,k_2}[ f]  \frac{d\sigma(k_2)}{|\nabla \mathscr H_0^k(k_2)|} - 2 \int_{\mathscr S'_k} R_{k+k_2,k,k_2} [ f]  \frac{d\sigma(k_2)}{|\nabla \mathscr H_1^k(k_2)|}.
	\end{equation}

	\section{Resonance Manifolds}\label{Sec:Energy}
	
	Our first step is to study the surface integrals.

	\begin{proposition}[Surface estimates on \(\mathscr S_{k_1}\) and
		\(\mathscr S'_{k_1}\)]\label{lem-Sp}
		Assume that \(d\geq 3\), \(1<\gamma< 2\), and
		\[
		\omega(k)=|k|^\gamma,
		\qquad
		\mathscr M(k)=|k|^{\gamma+d-2}\exp(-\mathscr A|k|),
		\qquad \mathscr A>0.
		\]
		
		Let \(F:\mathbb R^d\to\mathbb R\) be radial, namely
		\(
		F(k)=\mathcal F(|k|).
		\)
		Then there exist constants
		\(
		\mathcal C_S,\mathcal D_S,
		\mathcal C_{S'},\mathcal D_{S'}>0,
		\)
		depending only on \(d\), \(\gamma\), and \(\mathscr A\), such that the following
		estimates hold uniformly for all \(k_1\in\mathbb R^d\setminus\{0\}\).
		
		First,
		\[
		\begin{aligned}
			\left|
			\int_{\mathscr S_{k_1}}
			\mathrm d\sigma(k_2)\,
			\frac{
				\mathscr M(k_1-k_2)
				\mathscr M(k_2)
				\mathscr M(k_1)
				F(k_2)}
			{|\nabla_{k_2}\mathscr H_0^{k_1}(k_2)|}
			\right|                                           
			&\leq
			\mathcal C_S
			\int_0^{|k_1|} \mathrm dr\,
			|\mathcal F(r)|\,\omega(r)\,r^{d-1},
		\end{aligned}
		\]
		and
		\[
		\begin{aligned}
			\left|
			\int_{\mathscr S_{k_1}}
			\mathrm d\sigma(k_2)\,
			\frac{
				\mathscr M(k_1-k_2)
				\mathscr M(k_2)
				\mathscr M(k_1)
				F(k_2)}
			{|\nabla_{k_2}\mathscr H_0^{k_1}(k_2)|}
			\right|                                           
			&\leq
			\mathcal D_S
			\int_0^{|k_1|} \mathrm dr\,
			|\mathcal F(r)|\,r^{\gamma+d-1}\exp(-\mathscr A r).
		\end{aligned}
		\]
		Second,
		\[
		\begin{aligned}
			\left|
			\int_{\mathscr S'_{k_1}}
			\mathrm d\sigma(k_2)\,
			\frac{
				\mathscr M(k_1+k_2)
				\mathscr M(k_2)
				\mathscr M(k_1)
				F(k_2)}
			{|\nabla_{k_2}\mathscr H_1^{k_1}(k_2)|}
			\right|                                           
			&\leq
			\mathcal C_{S'}
			\int_0^\infty \mathrm dr\,
			|\mathcal F(r)|\,\omega(r)\,r^{d-1},
		\end{aligned}
		\]
		and
		\[
		\begin{aligned}
			\left|
			\int_{\mathscr S'_{k_1}}
			\mathrm d\sigma(k_2)\,
			\frac{
				\mathscr M(k_1+k_2)
				\mathscr M(k_2)
				\mathscr M(k_1)
				F(k_2)}
			{|\nabla_{k_2}\mathscr H_1^{k_1}(k_2)|}
			\right|                                           
			&\leq
			\mathcal D_{S'}
			\int_0^\infty \mathrm dr\,
			|\mathcal F(r)|\,r^{\gamma+d-1}\exp(-\mathscr A r).
		\end{aligned}
		\]
	\end{proposition}

	\begin{proof}
		We define
		\begin{equation}\label{def:A_B}
			\begin{aligned}
				A &:= \int_{\mathscr S_{k_1}} \mathrm d\sigma(k_2)\, \frac{\mathscr M(k_1-k_2)\, \mathscr M(k_2)\, \mathscr M(k_1)\, F(k_2)} {|\nabla \mathscr H_0^{k_1}|}, \\
				B &:= \int_{\mathscr S_{k_1}'} \mathrm d\sigma(k_2)\, \frac{\mathscr M(k_1+k_2)\, \mathscr M(k_2)\, \mathscr M(k_1)\, F(k_2)}	{|\nabla \mathscr H_1^{k_1}|}.
			\end{aligned}
		\end{equation}
		We then estimate \(A\) and \(B\) separately.
		
		{\it Step 1: Analyzing A.}
		
		\noindent Step 1 is divided into several substeps.
		
		\emph{Substep 1.1.} In this step, we will parametrize $\mathscr S_{k_1}$.

		Assume first that \(k_1\neq 0\), and set
		\[
		K:=|k_1|.
		\]
		Since the dispersion relation is radial, the surface
		\[
		\mathscr S_{k_1}
		:=
		\{v\in\mathbb R^d:
		\omega(k_1-v)+\omega(v)=\omega(k_1)\}
		\]
		is invariant under rotations fixing the axis spanned by \(k_1\).
		
		Let \(E_1,\ldots,E_{d-1}\) be an orthonormal basis of \(k_1^\perp\). For
		\[
		q=(q_1,\ldots,q_{d-1})\in S^{d-2}
		:=
		\{q\in\mathbb R^{d-1}: |q|=1\},
		\]
		we define
		\[
		e_q:=\sum_{j=1}^{d-1} q_j E_j .
		\]
		Then
		\(
		e_q\in k_1^\perp,
		|e_q|=1.
		\)
		We introduce the parametrization
		\[
		v=\xi k_1+\eta e_q,
		\qquad
		0<\xi<1,
		\qquad
		\eta\ge 0,
		\qquad
		q\in S^{d-2}.
		\]
		Since \(e_q\perp k_1\), we have
		\[
		|v|^2=\xi^2K^2+\eta^2,
		\qquad
		|k_1-v|^2=(1-\xi)^2K^2+\eta^2 .
		\]
		We define
		\[
		\mathscr H_0^{k_1}(\xi,\eta)
		:=
		\omega(k_1-v)+\omega(v)-\omega(k_1).
		\]
		Using \(\omega(k)=|k|^\gamma\), we obtain
		\[
		\mathscr H_0^{k_1}(\xi,\eta)
		=
		\bigl((1-\xi)^2K^2+\eta^2\bigr)^{\gamma/2}
		+
		\bigl(\xi^2K^2+\eta^2\bigr)^{\gamma/2}
		-
		K^\gamma .
		\]
		
		We now show that for every \(0<\xi<1\), there exists a unique
		\(\eta=\eta_K(\xi)>0\) such that
		\[
		\mathscr H_0^{k_1}(\xi,\eta_K(\xi))=0.
		\]
		Indeed, we have
		\[
		\mathscr H_0^{k_1}(\xi,0)
		=
		K^\gamma\bigl((1-\xi)^\gamma+\xi^\gamma-1\bigr)<0,
		\]
		because \(1<\gamma< 2\), while
		\(
		\mathscr H_0^{k_1}(\xi,\eta)\longrightarrow+\infty
		\mbox{ as }\eta\to+\infty .
		\)
		Thus the existence result follows by continuity.
		Now, we compute
		\[
		\partial_\eta \mathscr H_0^{k_1}(\xi,\eta)
		=
		\gamma\eta
		\left[
		\bigl((1-\xi)^2K^2+\eta^2\bigr)^{\gamma/2-1}
		+
		\bigl(\xi^2K^2+\eta^2\bigr)^{\gamma/2-1}
		\right].
		\]
		Hence, we conclude
		\(
		\partial_\eta \mathscr H_0^{k_1}(\xi,\eta)>0
		\mbox{ for every }\eta>0 .
		\)
		Therefore \(\mathscr H_0^{k_1}(\xi,\eta)\) is strictly increasing in
		\(\eta\), and the zero \(\eta_K(\xi)\) is unique.
		
		Since
		\(
		\partial_\eta \mathscr H_0^{k_1}(\xi,\eta_K(\xi))>0,
		\)
		the implicit function theorem implies that \(\eta_K\) is smooth for $\xi\in$ 
		\((0,1)\). Moreover, by homogeneity, we may write
		\[
		\eta_K(\xi)=K\rho_\gamma(\xi),
		\]
		where \(\rho_\gamma(\xi)>0\) is the unique solution of
		\[
		\bigl((1-\xi)^2+\rho_\gamma(\xi)^2\bigr)^{\gamma/2}
		+
		\bigl(\xi^2+\rho_\gamma(\xi)^2\bigr)^{\gamma/2}
		=
		1 .
		\]
		Consequently, we obtain
		\[
		\mathscr S_{k_1}\setminus\{0,k_1\}
		=
		\left\{
		\xi k_1+K\rho_\gamma(\xi)e_q:
		0<\xi<1,\ q\in S^{d-2}
		\right\}.
		\]

		\emph{Substep 1.2.}

		We next record two elementary properties of the function \(\eta_K\). Namely,
		\(\eta_K\) is symmetric with respect to \(\xi=1/2\), admits the continuous
		extension
		\(
		\eta_K(0)=\eta_K(1)=0,
		\)
		and is strictly increasing on \((0,1/2)\). Hence it is invertible on
		\((0,1/2)\).
		
		The symmetry follows from the identity
		\[
		\mathscr H_0^{k_1}(\xi,\eta)
		=
		\mathscr H_0^{k_1}(1-\xi,\eta).
		\]
		By uniqueness of the zero in the \(\eta\)-variable, we obtain
		\(
		\eta_K(\xi)=\eta_K(1-\xi),
		0<\xi<1.
		\)
		Moreover, from the defining equation
		\[
		\bigl((1-\xi)^2K^2+\eta_K(\xi)^2\bigr)^{\gamma/2}
		+
		\bigl(\xi^2K^2+\eta_K(\xi)^2\bigr)^{\gamma/2}
		=
		K^\gamma,
		\]
		one sees that \(\eta_K\) extends continuously to the endpoints with
		\(
		\eta_K(0)=\eta_K(1)=0.
		\)
		
		It remains to prove monotonicity on \((0,1/2)\). Let
		\[
		v_\xi:=\xi k_1+\eta_K(\xi)e_q\in \mathscr S_{k_1}.
		\]
		We also have
		\[
		\mathscr H_0^{k_1}(\xi,\eta_K(\xi))=0.
		\]
		Differentiating this identity with respect to \(\xi\), we obtain
		\[
		0=
		\partial_\xi\mathscr H_0^{k_1}(\xi,\eta_K(\xi))
		+
		\eta_K'(\xi)
		\partial_\eta\mathscr H_0^{k_1}(\xi,\eta_K(\xi)).
		\]
		Equivalently, we have
		\[
		\begin{aligned}
			0
			&=
			-K^2
			\frac{\omega'(|k_1-v_\xi|)}{|k_1-v_\xi|}
			+
			\left(K^2\xi+\eta_K(\xi)\eta_K'(\xi)\right)
			\left[
			\frac{\omega'(|k_1-v_\xi|)}{|k_1-v_\xi|}
			+
			\frac{\omega'(|v_\xi|)}{|v_\xi|}
			\right].
		\end{aligned}
		\]
		Therefore, we get
		\[
		\eta_K(\xi)\eta_K'(\xi)
		=
		K^2
		\frac{
			(1-\xi)\dfrac{\omega'(|k_1-v_\xi|)}{|k_1-v_\xi|}
			-
			\xi\dfrac{\omega'(|v_\xi|)}{|v_\xi|}
		}{
			\dfrac{\omega'(|k_1-v_\xi|)}{|k_1-v_\xi|}
			+
			\dfrac{\omega'(|v_\xi|)}{|v_\xi|}
		}.
		\]
		We now use the power-law form \(\omega(r)=r^\gamma\). We set
		\[
		r_1:=|k_1-v_\xi|,
		\qquad
		r_2:=|v_\xi|.
		\]
		Then, we find
		\[
		\frac{\omega'(r_1)}{r_1}
		=
		\gamma r_1^{\gamma-2},
		\qquad
		\frac{\omega'(r_2)}{r_2}
		=
		\gamma r_2^{\gamma-2}.
		\]
		For \(0<\xi<1/2\), we have \(r_1>r_2\). Moreover,
		\[
		\frac{r_1}{r_2}
		=
		\left(
		\frac{(1-\xi)^2K^2+\eta_K(\xi)^2}
		{\xi^2K^2+\eta_K(\xi)^2}
		\right)^{1/2}
		<
		\frac{1-\xi}{\xi}.
		\]
		Since \(1<\gamma\le 2\), we have \(0\le 2-\gamma<1\). Hence, we obtain
		\[
		\frac{
			\dfrac{\omega'(r_2)}{r_2}
		}{
			\dfrac{\omega'(r_1)}{r_1}
		}
		=
		\left(\frac{r_1}{r_2}\right)^{2-\gamma}
		<
		\frac{1-\xi}{\xi}.
		\]
		Thus, we have
		\[
		(1-\xi)\frac{\omega'(r_1)}{r_1}
		-
		\xi\frac{\omega'(r_2)}{r_2}
		>0.
		\]
		Consequently, we have
		\(
		\eta_K(\xi)\eta_K'(\xi)>0, 0<\xi<\frac12.
		\)
		Since \(\eta_K(\xi)>0\) for \(0<\xi<1\), it follows that
		\(
		\eta_K'(\xi)>0,
		0<\xi<\frac12.
		\)
		Therefore \(\eta_K\) is strictly increasing, and hence invertible, on
		\((0,1/2)\).

		\emph{Substep 1.3.} 
		
		We next compute the induced surface measure under the parametrization
		\[
		v(\xi,q)=\xi k_1+\eta_K(\xi)e_q,
		\qquad
		0<\xi<1,
		\qquad
		q\in S^{d-2}.
		\]
		Here \(e_q\in k_1^\perp\), \(|e_q|=1\), and \(\mathrm d\sigma_{d-2}(q)\)
		denotes the standard surface measure on \(S^{d-2}\). We again use the notations
		\(
		K:=|k_1|,
		v_\xi:=v(\xi,q),
		\)
		and introduce the abbreviations
		\[
		r_1:=|k_1-v_\xi|,
		\qquad
		r_2:=|v_\xi|, \qquad
		a:=\frac{\omega'(r_1)}{r_1},
		\qquad
		b:=\frac{\omega'(r_2)}{r_2}.
		\]
		Since
		\[
		|v_\xi|^2=\xi^2K^2+\eta_K(\xi)^2,
		\qquad
		|k_1-v_\xi|^2=(1-\xi)^2K^2+\eta_K(\xi)^2,
		\]
		the identity
		\[
		\mathscr H_0^{k_1}(\xi,\eta_K(\xi))=0
		\]
		is equivalent to
		\(
		\omega(r_1)+\omega(r_2)=\omega(K).
		\)
		
		Differentiating this identity with respect to \(\xi\), we obtain
		\[
		0
		=
		K^2\bigl[\xi b+(\xi-1)a\bigr]
		+
		\eta_K(\xi)\eta_K'(\xi)(a+b).
		\]
		From here, we find
		\[
		\frac12\partial_\xi\eta_K(\xi)^2
		=
		K^2
		\frac{(1-\xi)a-\xi b}{a+b}.
		\]
		Thus, we get
		\[
		\partial_\xi\eta_K(\xi)^2
		=
		2K^2
		\frac{(1-\xi)\dfrac{\omega'(r_1)}{r_1}
			-
			\xi\dfrac{\omega'(r_2)}{r_2}}
		{\dfrac{\omega'(r_1)}{r_1}
			+
			\dfrac{\omega'(r_2)}{r_2}}.
		\]
		
		We now compute the surface element. Since the angular directions are
		orthogonal to both \(k_1\) and \(e_q\), the induced metric gives
		\[
		\mathrm d\sigma(v)
		=
		\eta_K(\xi)^{d-2}
		\sqrt{K^2+\eta_K'(\xi)^2}\,
		\mathrm d\xi\,\mathrm d\sigma_{d-2}(q).
		\]
		Equivalently, we obtain
		\[
		\mathrm d\sigma(v)
		=
		\eta_K(\xi)^{d-3}
		\left(
		K^2\eta_K(\xi)^2
		+
		\frac14\left|\partial_\xi\eta_K(\xi)^2\right|^2
		\right)^{1/2}
		\mathrm d\xi\,\mathrm d\sigma_{d-2}(q).
		\]
		
		Next, we observe that
		\[
		\nabla_v\mathscr H_0^{k_1}
		=
		\frac{v_\xi-k_1}{|k_1-v_\xi|}
		\omega'(|k_1-v_\xi|)
		+
		\frac{v_\xi}{|v_\xi|}
		\omega'(|v_\xi|).
		\]
		Using
		\(
		v_\xi=\xi k_1+\eta_K(\xi)e_q,
		\)
		we get
		\[
		\nabla_v\mathscr H_0^{k_1}
		=
		\bigl[\xi b+(\xi-1)a\bigr]k_1
		+
		\eta_K(\xi)(a+b)e_q.
		\]
		Therefore, we find
		\[
		|\nabla_v\mathscr H_0^{k_1}|^2
		=
		K^2\bigl[\xi b+(\xi-1)a\bigr]^2
		+
		\eta_K(\xi)^2(a+b)^2.
		\]
		Using
		\[
		K^2\bigl[\xi b+(\xi-1)a\bigr]
		=
		-\eta_K(\xi)\eta_K'(\xi)(a+b),
		\]
		we obtain
		\[
		|\nabla_v\mathscr H_0^{k_1}|
		=
		\frac{
			\left(
			K^2\eta_K(\xi)^2
			+
			\frac14\left|\partial_\xi\eta_K(\xi)^2\right|^2
			\right)^{1/2}}
		{K}
		(a+b).
		\]
		Consequently, we have
		\[
		\frac{\mathrm d\sigma(v)}
		{|\nabla_v\mathscr H_0^{k_1}|}
		=
		\frac{K\eta_K(\xi)^{d-3}}{a+b}
		\mathrm d\xi\,\mathrm d\sigma_{d-2}(q).
		\]
		
		We now introduce the radial variable
		\[
		u:=|v_\xi|=r_2.
		\]
		Since
		\(
		u^2=\xi^2K^2+\eta_K(\xi)^2,
		\)
		we have
		\[
		\partial_\xi u^2
		=
		2\xi K^2+\partial_\xi\eta_K(\xi)^2.
		\]
		Using the formula for \(\partial_\xi\eta_K(\xi)^2\), we find
		\(
		\partial_\xi u^2
		=
		\frac{2K^2a}{a+b}.
		\)
		Hence
		\[
		\mathrm d\xi
		=
		\frac{a+b}{K^2a}\,u\,\mathrm du.
		\]
		Therefore, we have
		\[
		\frac{\mathrm d\sigma(v)}
		{|\nabla_v\mathscr H_0^{k_1}|}
		=
		\eta_K(\xi(u))^{d-3}
		\frac{u}{Ka}
		\mathrm du\,\mathrm d\sigma_{d-2}(q).
		\]
		Since \(a=\omega'(r_1)/r_1\), this becomes
		\[
		\frac{\mathrm d\sigma(v)}
		{|\nabla_v\mathscr H_0^{k_1}|}
		=
		\eta_K(\xi(u))^{d-3}
		\frac{u\,r_1(u)}
		{K\,\omega'(r_1(u))}
		\mathrm du\,\mathrm d\sigma_{d-2}(q).
		\]
		
		Thus, for a radial function \(F\), we obtain
		\[
		\begin{aligned}
			A
			&=
			\int_0^K \mathrm du
			\int_{S^{d-2}} \mathrm d\sigma_{d-2}(q)\,
			\eta_K(\xi(u))^{d-3}
			\frac{u\,r_1(u)}
			{K\,\omega'(r_1(u))}      
			\mathscr M(k_1-v(u,q))
			\mathscr M(v(u,q))
			\mathscr M(k_1)
			F(u),
		\end{aligned}
		\]
		where
		\[
		v(u,q)=\xi(u)k_1+\eta_K(\xi(u))e_q,
		\qquad
		u=|v(u,q)|,
		\]
		and
		\(
		r_1(u)=|k_1-v(u,q)|.
		\)
		Equivalently, since
		\(
		\omega(r_1(u))+\omega(u)=\omega(K),
		\)
		one has
		\(
		r_1(u)=\omega^{-1}\bigl(\omega(K)-\omega(u)\bigr).
		\)
		For the power law \(\omega(r)=r^\gamma\), this becomes
		\[
		r_1(u)=\bigl(K^\gamma-u^\gamma\bigr)^{1/\gamma}.
		\]

		\emph{Substep 1.4.} 
		
		We now estimate \(A\). Using the representation obtained in the previous
		substep, we have, for \(K=|k_1|>0\),
		\[
		\begin{aligned}
			|A|
			&\leq
			\int_0^K \mathrm du
			\int_{S^{d-2}} \mathrm d\sigma_{d-2}(q)\,
			\eta_K(\xi(u))^{d-3}
			\frac{u\,r_1(u)}
			{K\,\omega'(r_1(u))}
			\mathscr M(k_1-v(u,q))
			\mathscr M(v(u,q))
			\mathscr M(k_1)
			|F(u)| .
		\end{aligned}
		\]
		Here
		\(
		u=|v(u,q)|,
		r_1(u)=|k_1-v(u,q)|,
		\)
		and, on the resonant surface,
		\[
		\omega(r_1(u))+\omega(u)=\omega(K).
		\]
		Since \(\omega(r)=r^\gamma\), we have
		\[
		r_1(u)=\bigl(K^\gamma-u^\gamma\bigr)^{1/\gamma},
		\qquad
		0<u<K.
		\]
		Moreover,
		\(
		\omega'(r_1)=\gamma r_1^{\gamma-1}.
		\)
		We obtain
		\[
		\begin{aligned}
			|A|
			&\leq
			C_d
			\int_0^K \mathrm du\,
			\eta_K(\xi(u))^{d-3}
			r_1(u)^d
			K^{\gamma+d-3}
			u^{\gamma+d-1}
			\exp\bigl(-\mathscr A(r_1(u)+K+u)\bigr)
			|F(u)| .
		\end{aligned}
		\]
		Since \(0\leq u\leq K\), \(0\leq r_1(u)\leq K\), and
		\(0\leq \eta_K(\xi(u))\leq K\), we have, for \(d\geq 3\),
		\[
		\eta_K(\xi(u))^{d-3}
		r_1(u)^d
		K^{\gamma+d-3}
		\exp\bigl(-\mathscr A(r_1(u)+K)\bigr)
		\leq
		K^{\gamma+3d-6}\exp(-\mathscr A K)
		\leq C .
		\]
		Consequently, we get
		\[
		\begin{aligned}
			|A|
			&\leq
			\mathcal D_S
			\int_0^K \mathrm du\,
			|F(u)|\,u^{\gamma+d-1}\exp(-\mathscr A u).
		\end{aligned}
		\]
		In particular, since \(\exp(-\mathscr A u)\leq 1\), we bound
		\[
		\begin{aligned}
			|A|
			\leq
			\mathcal C_S
			\int_0^K \mathrm du\,
			|F(u)|\,u^{\gamma+d-1}        \ 
			&=
			\mathcal C_S
			\int_0^K \mathrm du\,
			|F(u)|\,\omega(u)\,u^{d-1}.
		\end{aligned}
		\]
		Thus, we have
		\[
		|A|
		\leq
		\mathcal C_S
		\int_0^{|k_1|} \mathrm du\,
		|F(u)|\,\omega(u)\,u^{d-1},
		\]
		and also
		\[
		|A|
		\leq
		\mathcal D_S
		\int_0^{|k_1|} \mathrm du\,
		|F(u)|\,u^{\gamma+d-1}\exp(-\mathscr A u).
		\]

		{\it Step 2: Analyzing B.}
		
		\noindent We also divide Step 2 into smaller substeps.
		
		\emph{Substep 2.1.}
		We now consider the second resonant surface
		\[
		\mathscr S'_{k_1}
		:=
		\{v\in\mathbb R^d:
		\omega(k_1+v)-\omega(k_1)-\omega(v)=0\}.
		\]
		Equivalently, up to a harmless sign in the defining function, we set
		\[
		\mathscr H_1^{k_1}(v)
		:=
		\omega(k_1+v)-\omega(k_1)-\omega(v).
		\]
		Then, we get
		\[
		\mathscr S'_{k_1}
		=
		\{v\in\mathbb R^d:\mathscr H_1^{k_1}(v)=0\}.
		\]
		We assume first that \(k_1\neq 0\), and set
		\[
		K:=|k_1|.
		\]
		
		We claim that, for \(1<\gamma<2\), the nontrivial part of
		\(\mathscr S'_{k_1}\) can be parametrized by
		\[
		v=\xi k_1+\eta_K(\xi)e_q,
		\qquad
		0<\xi<\infty,
		\qquad
		q\in S^{d-2},
		\]
		where \(e_q\in k_1^\perp\), \(|e_q|=1\), and where
		\(\eta_K(\xi)>0\) is uniquely determined by the resonance condition.
		
		We first explain why only \(\xi>0\) occurs. We write, more generally,
		\[
		v=\zeta k_1+\eta e_q,
		\qquad
		\zeta\in\mathbb R,
		\qquad
		\eta\ge 0.
		\]
		Then
		\[
		|v|^2=\zeta^2K^2+\eta^2,
		\qquad
		|k_1+v|^2=(1+\zeta)^2K^2+\eta^2.
		\]
		If \(\zeta\le -1/2\), then
		\(
		|k_1+v|\le |v|,
		\)
		and hence
		\(
		\omega(k_1+v)-\omega(k_1)-\omega(v)<0.
		\)
		Thus no point with \(\zeta\le -1/2\) lies on the nontrivial resonant surface.
		
		Next suppose that \(-1/2<\zeta\le 0\). We define
		\[
		\Psi_\zeta(\eta)
		:=
		\bigl((1+\zeta)^2K^2+\eta^2\bigr)^{\gamma/2}
		-
		K^\gamma
		-
		\bigl(\zeta^2K^2+\eta^2\bigr)^{\gamma/2}.
		\]
		At \(\eta=0\), we have
		\[
		\Psi_\zeta(0)
		=
		K^\gamma\bigl((1+\zeta)^\gamma-1-|\zeta|^\gamma\bigr)
		\le 0,
		\]
		with equality only at the endpoint \(\zeta=0\). Moreover,
		for \(\eta>0\),
		\[
		\partial_\eta\Psi_\zeta(\eta)
		=
		\gamma\eta
		\left[
		\bigl((1+\zeta)^2K^2+\eta^2\bigr)^{\gamma/2-1}
		-
		\bigl(\zeta^2K^2+\eta^2\bigr)^{\gamma/2-1}
		\right]
		\le 0,
		\]
		because
		\(
		(1+\zeta)^2K^2+\eta^2
		\ge
		\zeta^2K^2+\eta^2
		\)
		and
		\(
		\frac{\gamma}{2}-1<0.
		\)
		Hence \(\Psi_\zeta(\eta)<0\) for \(\eta>0\). Therefore, apart from the
		degenerate endpoint \(v=0\), no point with \(\zeta\le 0\) belongs to
		\(\mathscr S'_{k_1}\). Thus the nontrivial surface is parametrized by
		\(\xi>0\).
		
		We now fix \(\xi>0\) and prove the existence and uniqueness of
		\(\eta_K(\xi)>0\). We define
		\[
		\Phi_\xi(\eta)
		:=
		\bigl((1+\xi)^2K^2+\eta^2\bigr)^{\gamma/2}
		-
		K^\gamma
		-
		\bigl(\xi^2K^2+\eta^2\bigr)^{\gamma/2}.
		\]
		The resonance condition
		\(
		\omega(k_1+v)=\omega(k_1)+\omega(v)
		\)
		is equivalent to
		\(
		\Phi_\xi(\eta)=0.
		\)
		At \(\eta=0\), we have
		\[
		\Phi_\xi(0)
		=
		K^\gamma\bigl((1+\xi)^\gamma-1-\xi^\gamma\bigr).
		\]
		Since \(1<\gamma<2\), the function
		\[
		\xi\mapsto (1+\xi)^\gamma-\xi^\gamma
		\]
		is strictly increasing on \((0,\infty)\) and has value \(1\) at
		\(\xi=0\). Hence, we have
		\[
		(1+\xi)^\gamma-1-\xi^\gamma>0,
		\]
		and therefore
		\(
		\Phi_\xi(0)>0.
		\)
		
		On the other hand, since \(1<\gamma<2\),
		\[
		\bigl((1+\xi)^2K^2+\eta^2\bigr)^{\gamma/2}
		-
		\bigl(\xi^2K^2+\eta^2\bigr)^{\gamma/2}
		\longrightarrow 0
		\qquad
		\text{as } \eta\to+\infty.
		\]
		Consequently, we obtain
		\[
		\lim_{\eta\to+\infty}\Phi_\xi(\eta)
		=
		-K^\gamma<0.
		\]
		By the continuity of the function, there exists at least one positive zero of
		\(\Phi_\xi\).
		
		We compute
		\[
		\partial_\eta\Phi_\xi(\eta)
		=
		\gamma\eta
		\left[
		\bigl((1+\xi)^2K^2+\eta^2\bigr)^{\gamma/2-1}
		-
		\bigl(\xi^2K^2+\eta^2\bigr)^{\gamma/2-1}
		\right].
		\]
		Since
		\(
		(1+\xi)^2K^2+\eta^2
		>
		\xi^2K^2+\eta^2
		\)
		and
		\(
		\frac{\gamma}{2}-1<0,
		\)
		we have
		\[
		\bigl((1+\xi)^2K^2+\eta^2\bigr)^{\gamma/2-1}
		<
		\bigl(\xi^2K^2+\eta^2\bigr)^{\gamma/2-1}.
		\]
		Therefore
		\(
		\partial_\eta\Phi_\xi(\eta)<0
		\text{for every } \eta>0.
		\)
		Thus \(\eta\mapsto \Phi_\xi(\eta)\) is strictly decreasing on
		\((0,\infty)\). Since it starts positive at \(\eta=0\) and tends to a
		negative limit as \(\eta\to+\infty\), it has exactly one positive zero.
		We denote this zero by
		\[
		\eta=\eta_K(\xi).
		\]
		Moreover, we have
		\[
		\partial_\eta\Phi_\xi(\eta_K(\xi))<0,
		\]
		so the implicit function theorem implies that
		\(
		\xi\mapsto \eta_K(\xi)
		\)
		is smooth on \((0,\infty)\).
		
		By homogeneity, we may write
		\(
		\eta_K(\xi)=K\rho_\gamma(\xi),
		\)
		where \(\rho_\gamma(\xi)>0\) is the unique solution of
		\[
		\bigl((1+\xi)^2+\rho_\gamma(\xi)^2\bigr)^{\gamma/2}
		=
		1+
		\bigl(\xi^2+\rho_\gamma(\xi)^2\bigr)^{\gamma/2}.
		\]
		Consequently,
		we have \[
		\mathscr S'_{k_1}\setminus\{0\}
		=
		\left\{
		\xi k_1+K\rho_\gamma(\xi)e_q:
		0<\xi<\infty,\ q\in S^{d-2}
		\right\}.
		\]
		
		We now compute the surface element in this parametrization. We set
		\[
		v_\xi:=\xi k_1+\eta_K(\xi)e_q,
		\qquad
		r_+:=|k_1+v_\xi|,
		\qquad
		r:=|v_\xi|.
		\]
		We also introduce
		\[
		a:=\frac{\omega'(r_+)}{r_+},
		\qquad
		b:=\frac{\omega'(r)}{r}.
		\]
		On \(\mathscr S'_{k_1}\), one has
		\[
		\omega(r_+)=\omega(K)+\omega(r).
		\]
		Since \(\omega(r)=r^\gamma\), we have
		\(
		a=\gamma r_+^{\gamma-2},
		b=\gamma r^{\gamma-2}.
		\)
		Moreover \(r_+>r\), and since \(\gamma-2<0\), it follows that
		\(
		0<a<b.
		\)
		In particular,
		\(
		b-a>0.
		\)
		
		Differentiating the identity
		\(
		\mathscr H_1^{k_1}(v_\xi)=0
		\)
		with respect to \(\xi\), we obtain
		\[
		0
		=
		K^2\bigl[(1+\xi)a-\xi b\bigr]
		+
		\eta_K(\xi)\eta_K'(\xi)(a-b).
		\]
		Equivalently, we have
		\[
		\eta_K(\xi)\eta_K'(\xi)
		=
		K^2
		\frac{(1+\xi)a-\xi b}{b-a}.
		\]
		
		We next compute the surface element. The induced surface measure is
		\[
		\mathrm d\sigma(v)
		=
		\eta_K(\xi)^{d-2}
		\sqrt{K^2+\eta_K'(\xi)^2}\,
		\mathrm d\xi\,\mathrm d\sigma_{d-2}(q).
		\]
		Equivalently, we have
		\[
		\mathrm d\sigma(v)
		=
		\eta_K(\xi)^{d-3}
		\left(
		K^2\eta_K(\xi)^2
		+
		\frac14
		\left|\partial_\xi \eta_K(\xi)^2\right|^2
		\right)^{1/2}
		\mathrm d\xi\,\mathrm d\sigma_{d-2}(q).
		\]
		
		Moreover, we observe that
		\[
		\nabla_v\mathscr H_1^{k_1}(v_\xi)
		=
		\frac{k_1+v_\xi}{|k_1+v_\xi|}
		\omega'(|k_1+v_\xi|)
		-
		\frac{v_\xi}{|v_\xi|}
		\omega'(|v_\xi|).
		\]
		Using
		\[
		v_\xi=\xi k_1+\eta_K(\xi)e_q,
		\]
		this becomes
		\[
		\nabla_v\mathscr H_1^{k_1}(v_\xi)
		=
		\bigl[(1+\xi)a-\xi b\bigr]k_1
		+
		\eta_K(\xi)(a-b)e_q.
		\]
		Therefore, we find
		\[
		|\nabla_v\mathscr H_1^{k_1}(v_\xi)|^2
		=
		K^2\bigl[(1+\xi)a-\xi b\bigr]^2
		+
		\eta_K(\xi)^2(a-b)^2.
		\]
		Using
		\[
		K^2\bigl[(1+\xi)a-\xi b\bigr]
		=
		-\eta_K(\xi)\eta_K'(\xi)(a-b),
		\]
		we obtain
		\[
		|\nabla_v\mathscr H_1^{k_1}(v_\xi)|
		=
		\frac{
			\left(
			K^2\eta_K(\xi)^2
			+
			\frac14
			\left|\partial_\xi\eta_K(\xi)^2\right|^2
			\right)^{1/2}}
		{K}
		\,|b-a|.
		\]
		Since \(b-a>0\), this gives
		\[
		\frac{\mathrm d\sigma(v)}
		{|\nabla_v\mathscr H_1^{k_1}(v)|}
		=
		\frac{K\eta_K(\xi)^{d-3}}{b-a}
		\mathrm d\xi\,\mathrm d\sigma_{d-2}(q).
		\]
		
		We now introduce the radial variable
		\[
		u:=|v_\xi|=r.
		\]
		Since
		\(
		u^2=\xi^2K^2+\eta_K(\xi)^2,
		\)
		we have
		\[
		\partial_\xi u^2
		=
		2\xi K^2+\partial_\xi\eta_K(\xi)^2.
		\]
		Using the differentiated resonance equation, we obtain
		\[
		\partial_\xi u^2
		=
		\frac{2K^2a}{b-a}.
		\]
		In particular,
		\(
		\partial_\xi u^2>0.
		\)
		Thus the map \(\xi\mapsto u\) is strictly increasing. Moreover,
		\(u\to0\) as \(\xi\downarrow0\), while \(u\to+\infty\) as
		\(\xi\to+\infty\). Hence \(u\) gives a valid change of variables from
		\((0,\infty)_\xi\) onto \((0,\infty)_u\).
		
		Since
		\[
		2u\,\mathrm du
		=
		\frac{2K^2a}{b-a}\,\mathrm d\xi,
		\]
		we have
		\[
		\mathrm d\xi
		=
		\frac{b-a}{K^2a}\,u\,\mathrm du.
		\]
		Consequently, we find
		\[
		\frac{\mathrm d\sigma(v)}
		{|\nabla_v\mathscr H_1^{k_1}(v)|}
		=
		\eta_K(\xi(u))^{d-3}
		\frac{u}{Ka}
		\mathrm du\,\mathrm d\sigma_{d-2}(q).
		\]
		Since
		\(
		a=\frac{\omega'(r_+)}{r_+},
		\)
		this becomes
		\[
		\frac{\mathrm d\sigma(v)}
		{|\nabla_v\mathscr H_1^{k_1}(v)|}
		=
		\eta_K(\xi(u))^{d-3}
		\frac{u\,r_+(u)}
		{K\,\omega'(r_+(u))}
		\mathrm du\,\mathrm d\sigma_{d-2}(q).
		\]
		
		Therefore, for a radial function \(F(k)=\mathcal F(|k|)\), we obtain
		\[
		\begin{aligned}
			B
			&=
			\int_0^\infty \mathrm du
			\int_{S^{d-2}} \mathrm d\sigma_{d-2}(q)\,
			\eta_K(\xi(u))^{d-3}
			\frac{u\,r_+(u)}
			{K\,\omega'(r_+(u))}
			\\
			&\qquad\qquad\times
			\mathscr M(k_1+v(u,q))
			\mathscr M(v(u,q))
			\mathscr M(k_1)
			\mathcal F(u),
		\end{aligned}
		\]
		where
		\[
		v(u,q)=\xi(u)k_1+\eta_K(\xi(u))e_q,
		\qquad
		u=|v(u,q)|,
		\]
		and
		\(
		r_+(u):=|k_1+v(u,q)|.
		\)
		On the resonant surface, we have
		\[
		\omega(r_+(u))=\omega(K)+\omega(u).
		\]
		Thus, we obtain
		\[
		r_+(u)=\omega^{-1}\bigl(\omega(K)+\omega(u)\bigr).
		\]
		For the power law \(\omega(r)=r^\gamma\), this becomes
		\[
		r_+(u)=\bigl(K^\gamma+u^\gamma\bigr)^{1/\gamma}.
		\]

		\emph{Substep 2.2.} 	
		We now estimate \(B\). From the representation obtained in the previous
		substep, with \(K=|k_1|>0\), we have
		\[
		\begin{aligned}
			|B|
			&\leq
			\int_0^\infty \mathrm du
			\int_{S^{d-2}} \mathrm d\sigma_{d-2}(q)\,
			\eta_K(\xi(u))^{d-3}
			\frac{u\,r_+(u)}
			{K\,\omega'(r_+(u))}        \\
			&\qquad\qquad\times
			\mathscr M(k_1+v(u,q))
			\mathscr M(v(u,q))
			\mathscr M(k_1)
			|F(u)|.
		\end{aligned}
		\]
		
		On the resonant surface,
		\(
		\omega(r_+(u))=\omega(K)+\omega(u).
		\)
		Since \(\omega(r)=r^\gamma\), this gives
		\(
		r_+(u)=\bigl(K^\gamma+u^\gamma\bigr)^{1/\gamma}.
		\)
		Moreover,
		we also have \(
		\omega'(r_+)=\gamma r_+^{\gamma-1}.
		\)
		Using
		\[
		\mathscr M(k)=|k|^{\gamma+d-2}\exp(-\mathscr A |k|),
		\]
		we obtain
		\[
		\begin{aligned}
			|B|
			&\leq
			C_d
			\int_0^\infty \mathrm du\,
			\eta_K(\xi(u))^{d-3}
			r_+(u)^d
			K^{\gamma+d-3}
			u^{\gamma+d-1}      
			\exp\bigl(-\mathscr A (r_+(u)+K+u)\bigr)
			|F(u)| .
		\end{aligned}
		\]
		For \(d\geq 3\), we have
		\[
		0\leq \eta_K(\xi(u))\leq u\leq r_+(u),
		\qquad
		K\leq r_+(u).
		\]
		Therefore, we obtain
		\[
		\eta_K(\xi(u))^{d-3}
		r_+(u)^d
		K^{\gamma+d-3}
		\exp\bigl(-\mathscr A(r_+(u)+K)\bigr)
		\leq C,
		\]
		where \(C\) depends only on \(d\), \(\gamma\), and \(\mathscr A\). Consequently, we have
		\[
		|B|
		\leq
		\mathcal D_{S'}
		\int_0^\infty
		|F(u)|\,u^{\gamma+d-1}\exp(-\mathscr A u)\,\mathrm du .
		\]
		Since \(\exp(-\mathscr A u)\leq 1\), we also obtain
		\[
		|B|
		\leq
		\mathcal C_{S'}
		\int_0^\infty
		|F(u)|\,u^{\gamma+d-1}\,\mathrm du .
		\]
		Equivalently, because \(\omega(u)=u^\gamma\), we find
		\[
		|B|
		\leq
		\mathcal C_{S'}
		\int_0^\infty
		|F(u)|\,\omega(u)\,u^{d-1}\,\mathrm du .
		\]
		
	\end{proof}

	\section{The collision operator as a bilinear map}
	
	We set
	\[
	X:=L^\infty_{\rm rad}(\mathbb R^d)
	\]
	with norm
	\(
	\|f\|_X:=\|f\|_{L^\infty(\mathbb R^d)}.
	\)
	For \(f,g\in X\), we define the bilinear operator
	\[
	\mathcal B(f,g)(k):=\mathcal B_0(f,g)(k)-2\mathcal B_1(f,g)(k),
	\]
	where
	\[
	\begin{aligned}
		\mathcal B_0(f,g)(k)
		:=
		\int_{\mathscr S_k}
		\frac{
			V_{k,k_1,k-k_1}
		}{
			|\nabla \mathscr H_0^k(k_1)|
		}
		\Big[
		f(k_1)g(k-k_1)
		-f(k)g(k_1)
		-f(k)g(k-k_1)
		\Big]\,d\sigma(k_1),
	\end{aligned}
	\]
	and
	\[
	\begin{aligned}
		\mathcal B_1(f,g)(k)
		:=
		\int_{\mathscr S'_k}
		\frac{
			V_{k+k_2,k,k_2}
		}{
			|\nabla \mathscr H_1^k(k_2)|
		}
		\Big[
		f(k)g(k_2)
		-f(k)g(k+k_2)
		-f(k+k_2)g(k_2)
		\Big]\,d\sigma(k_2).
	\end{aligned}
	\]
	Then, we can write
	\[
	Q_{3w}[f]=\mathcal B(f,f).
	\]

	\begin{lemma}
		\label{lemma:bilinear-bound}
		There exists a constant \(C_B>0\), depending only on
		\(d,\gamma,\mathscr A\) and on the constants in
		Proposition~\ref{lem-Sp}, such that for all \(f,g\in X\),
		\[
		\|\mathcal B(f,g)\|_X
		\leq
		C_B\|f\|_X\|g\|_X.
		\]
		Consequently,
		\[
		\|Q_{3w}[f]-Q_{3w}[g]\|_X
		\leq
		C_B\bigl(\|f\|_X+\|g\|_X\bigr)\|f-g\|_X
		\]
		for all \(f,g\in X\).
	\end{lemma}
	
	\begin{proof}
		We first observe that \(\mathcal B(f,g)\) is radial whenever \(f\) and \(g\)
		are radial. Indeed, for every rotation \(R\in SO(d)\), we have
		\[
		\mathscr S_{Rk}=R\mathscr S_k,
		\qquad
		\mathscr S'_{Rk}=R\mathscr S'_k,
		\]
		and the quantities
		\(
		V_{k,k_1,k_2},
		|\nabla\mathscr H_0^k|,
		|\nabla\mathscr H_1^k|
		\)
		are invariant under simultaneous rotations of their arguments. A change of
		variables on the corresponding resonance surface therefore gives
		\[
		\mathcal B(f,g)(Rk)=\mathcal B(f,g)(k).
		\]
		Thus \(\mathcal B(f,g)\in X\), provided the integrals are bounded.
		
		We now estimate the \(S_k\)-part. Since \(f,g\in L^\infty\), we find
		\[
		\begin{aligned}
			|\mathcal B_0(f,g)(k)|
			&\leq
			3\|f\|_X\|g\|_X
			\int_{\mathscr S_k}
			\frac{
				V_{k,k_1,k-k_1}
			}{
				|\nabla\mathscr H_0^k(k_1)|
			}
			\,d\sigma(k_1).
		\end{aligned}
		\]
		Applying Proposition~\ref{lem-Sp} with the radial function \(F\equiv1\), and
		using the symmetry \(k_1\leftrightarrow k-k_1\) on \(\mathscr S_k\), we obtain
		\[
		\int_{\mathscr S_k}
		\frac{
			V_{k,k_1,k-k_1}
		}{
			|\nabla\mathscr H_0^k(k_1)|
		}
		\,d\sigma(k_1)
		\leq
		\mathcal D_S
		\int_0^{|k|}
		r^{\gamma+d-1}e^{-\mathscr A r}\,dr
		\leq
		\mathcal D_S G,
		\]
		where
		\[
		G:=\int_0^\infty r^{\gamma+d-1}e^{-\mathscr A r}\,dr<\infty.
		\]
		Therefore, we can bound
		\[
		|\mathcal B_0(f,g)(k)|
		\leq
		3\mathcal D_SG\|f\|_X\|g\|_X.
		\]
		
		Similarly, for the \(S'_k\)-part, we have
		\[
		\begin{aligned}
			|\mathcal B_1(f,g)(k)|
			&\leq
			3\|f\|_X\|g\|_X
			\int_{\mathscr S'_k}
			\frac{
				V_{k+k_2,k,k_2}
			}{
				|\nabla\mathscr H_1^k(k_2)|
			}
			\,d\sigma(k_2).
		\end{aligned}
		\]
		Again by Proposition~\ref{lem-Sp}, with \(F\equiv1\), we get
		\[
		\int_{\mathscr S'_k}
		\frac{
			V_{k+k_2,k,k_2}
		}{
			|\nabla\mathscr H_1^k(k_2)|
		}
		\,d\sigma(k_2)
		\leq
		\mathcal D_{S'}G.
		\]
		Hence, we obtain
		\[
		|\mathcal B_1(f,g)(k)|
		\leq
		3\mathcal D_{S'}G\|f\|_X\|g\|_X.
		\]
		Combining the estimates and recalling the factor \(2\) in front of
		\(\mathcal B_1\), we get
		\[
		\|\mathcal B(f,g)\|_X
		\leq
		3G(\mathcal D_S+2\mathcal D_{S'})
		\|f\|_X\|g\|_X.
		\]
		Thus one may take
		\(
		C_B:=3G(\mathcal D_S+2\mathcal D_{S'}).
		\)
		
		Finally, since
		\[
		Q_{3w}[f]-Q_{3w}[g]
		=
		\mathcal B(f,f)-\mathcal B(g,g)
		=
		\mathcal B(f-g,f)+\mathcal B(g,f-g),
		\]
		the Lipschitz estimate follows immediately from the bilinear bound.
	\end{proof}

	\section{Global time-analytic solutions of 3-wave kinetic equations}\label{Sec:3wave}
	
	We construct an exact strong solution of the 3-wave kinetic equation \eqref{3wave} in the form
	\begin{equation}\label{Sec:3wave:1}
		f(t,k) \;=\; \sum_{n=0}^{\infty} t^n \mathcal{A}_n(k),
	\end{equation}
	where $\mathcal{A}_0(k) = f_0(k)$.
	
	\noindent Plugging \eqref{Sec:3wave:1} into the 3-wave kinetic equation \eqref{3wave}, we find, by a direct computation
	\begin{equation}\label{Sec:3wave:2}
		\begin{aligned}
			\mathcal{A}_{n+1}(k) = {} & \frac{1}{n+1}\iint_{\mathbb{R}^{2d}} \mathrm{d}k_1 \, \mathrm{d}k_2 \, V_{k,k_1,k_2}\, \delta(k - k_1 - k_2)\, \delta(\omega_k - \omega_{k_1} - \omega_{k_2}) \\
			& \qquad \times \sum_{\substack{0 \le i,j \\ i+j = n}}\Big(	\mathcal{A}_i(k_1)\mathcal{A}_j(k_2) - \mathcal{A}_i(k)\mathcal{A}_j(k_1) - \mathcal{A}_i(k)\mathcal{A}_j(k_2)\Big) \\
			& - \frac{2}{n+1} \iint_{\mathbb{R}^{2d}} \mathrm{d}k_1 \, \mathrm{d}k_2 \, V_{k_1,k,k_2}\,	\delta(k_1 - k - k_2)\, \delta(\omega_{k_1} - \omega_k - \omega_{k_2}) \\
			& \qquad \times \sum_{\substack{0 \le i,j \\ i+j = n}} \Big(		\mathcal{A}_i(k)\mathcal{A}_j(k_2) - \mathcal{A}_i(k)\mathcal{A}_j(k_1) - \mathcal{A}_i(k_1)\mathcal{A}_j(k_2) \Big).
		\end{aligned}
	\end{equation}
	
	\begin{lemma}\label{Lemma:3Induction}
		Assume that
		\(
		M:=\sup_{k\in\mathbb R^d}|f_0(k)|<\infty
		\)
		and set
		\[
		G:=\int_0^\infty r^{\gamma+d-1}\exp(-\mathscr A r)\,\mathrm dr<\infty .	\]
		Then there exists a positive constant \(\mathcal C_1\), depending only on
		\(M\), \(G\), and the constants in Proposition~\ref{lem-Sp}, such that
		for every \(n=0,1,2,\ldots\),
		\[
		|\mathcal A_n(k)|
		\leq
		\mathcal C_1^n M
		\]
		for a.e. \(k\in\mathbb R^d\).
	\end{lemma}
	
	\begin{proof}
		We argue by  induction. The estimate is immediate for \(n=0\), since
		\(\mathcal A_0=f_0\). Assume that, for some \(n\geq 0\),
		\[
		|\mathcal A_\ell(k)|
		\leq
		\mathcal C_1^\ell M
		\]
		for every \(0\leq \ell\leq n\) and for a.e. \(k\in\mathbb R^d\). We prove
		the estimate for \(\mathcal A_{n+1}\).
		
		We fix indices \(i,j\geq 0\) with \(i+j=n\) and  define
		\[
		\begin{aligned}
			\mathcal A^1_{i,j}(k)
			&:=
			\iint_{\mathbb R^{2d}} \mathrm d k_1\,\mathrm d k_2\,
			V_{k,k_1,k_2}
			\delta(k-k_1-k_2)
			\delta(\omega_k-\omega_{k_1}-\omega_{k_2})        \\
			&\qquad\qquad\times
			\Big(
			\mathcal A_i(k_1)\mathcal A_j(k_2)
			-
			\mathcal A_i(k)\mathcal A_j(k_1)
			-
			\mathcal A_i(k)\mathcal A_j(k_2)
			\Big).
		\end{aligned}
		\]
		Using the delta constraint \(k_2=k-k_1\), we may write
		\[
		\begin{aligned}
			|\mathcal A^1_{i,j}(k)|
			&\leq
			\int_{\mathscr S_k}
			\frac{\mathrm d\sigma(k_1)}
			{|\nabla\mathscr H_0^k(k_1)|}\,
			V_{k,k_1,k-k_1}     
			\Big(
			|\mathcal A_i(k_1)|\,|\mathcal A_j(k-k_1)|
			+
			|\mathcal A_i(k)|\,|\mathcal A_j(k_1)|
			+
			|\mathcal A_i(k)|\,|\mathcal A_j(k-k_1)|
			\Big).
		\end{aligned}
		\]
		By the induction hypothesis and \(i+j=n\), we have
		\[
		|\mathcal A_i(\cdot)|\,|\mathcal A_j(\cdot)|
		\leq
		\mathcal C_1^n M^2 .
		\]
		Therefore, we find
		\[
		|\mathcal A^1_{i,j}(k)|
		\leq
		3\mathcal C_1^n M^2
		\int_{\mathscr S_k}
		\frac{\mathrm d\sigma(k_1)}
		{|\nabla\mathscr H_0^k(k_1)|}\,
		V_{k,k_1,k-k_1}.
		\]
		Applying Proposition~\ref{lem-Sp} gives
		\[
		|\mathcal A^1_{i,j}(k)|
		\leq
		3\mathcal D_S\,\mathcal C_1^n M^2 G .
		\]
		
		Next, we define
		\[
		\begin{aligned}
			\mathcal A^2_{i,j}(k)
			&:=
			-
			\iint_{\mathbb R^{2d}} \mathrm d k_1\,\mathrm d k_2\,
			V_{k_1,k,k_2}
			\delta(k_1-k-k_2)
			\delta(\omega_{k_1}-\omega_k-\omega_{k_2})        \\
			&\qquad\qquad\times
			\Big(
			\mathcal A_i(k)\mathcal A_j(k_2)
			-
			\mathcal A_i(k)\mathcal A_j(k_1)
			-
			\mathcal A_i(k_1)\mathcal A_j(k_2)
			\Big).
		\end{aligned}
		\]
		Using the constraint \(k_1=k+k_2\), we get
		\[
		\begin{aligned}
			|\mathcal A^2_{i,j}(k)|
			&\leq
			\int_{\mathscr S'_k}
			\frac{\mathrm d\sigma(k_2)}
			{|\nabla\mathscr H_1^k(k_2)|}\,
			V_{k+k_2,k,k_2}               
			\Big(
			|\mathcal A_i(k)|\,|\mathcal A_j(k_2)|
			+
			|\mathcal A_i(k)|\,|\mathcal A_j(k+k_2)|
			+
			|\mathcal A_i(k+k_2)|\,|\mathcal A_j(k_2)|
			\Big).
		\end{aligned}
		\]
		Again, by the induction hypothesis,
		\[
		|\mathcal A^2_{i,j}(k)|
		\leq
		3\mathcal C_1^n M^2
		\int_{\mathscr S'_k}
		\frac{\mathrm d\sigma(k_2)}
		{|\nabla\mathscr H_1^k(k_2)|}\,
		V_{k+k_2,k,k_2}.
		\]
		Applying Proposition~\ref{lem-Sp},  yields
		\[
		|\mathcal A^2_{i,j}(k)|
		\leq
		3\mathcal D_{S'}\,\mathcal C_1^n M^2 G .
		\]
		
		We recall the recurrence relation for \(\mathcal A_{n+1}\)
		\[
		\mathcal A_{n+1}(k)
		=
		\frac{1}{n+1}
		\sum_{\substack{i,j\geq 0\\ i+j=n}}
		\mathcal A^1_{i,j}(k)
		+
		\frac{2}{n+1}
		\sum_{\substack{i,j\geq 0\\ i+j=n}}
		\mathcal A^2_{i,j}(k).
		\]
		Since there are \(n+1\) pairs \((i,j)\) with \(i+j=n\), we obtain
		\[
		\begin{aligned}
			|\mathcal A_{n+1}(k)|
			&\leq
			3\mathcal D_S\,\mathcal C_1^n M^2 G
			+
			6\mathcal D_{S'}\,\mathcal C_1^n M^2 G     \leq
			9\mathcal C\,\mathcal C_1^n M^2 G,
		\end{aligned}
		\]
		where
		\(
		\mathcal C:=\max\{\mathcal D_S,\mathcal D_{S'}\}.
		\)
		We choose
		{\(
			\mathcal C_1:=9\mathcal C G M.
			\)}
		Then
		\[
		|\mathcal A_{n+1}(k)|
		\leq
		\mathcal C_1^{n+1}M .
		\]
		This closes the induction and proves the lemma.
	\end{proof}

	From \eqref{Sec:3wave:1}, we obtain the bound
	\begin{equation}\label{Sec:3wave:3}
		\big| f(t,k) \big| \;\le\; \sum_{n=0}^{\infty} t^n \mathcal C_1^{\,n}\, M \;\le\; \frac{M}{1 - t \mathcal C_1},
	\end{equation}
	for a.e.\ $k\in\mathbb{R}^d$ and all $0 \leq  t < \frac{1}{\mathcal C_1}$. We conclude that $f$ is analytic in $\left[0,\dfrac{1}{\mathcal C_1}\right)$.

	\begin{lemma}\label{Lemma:local-uniqueness} We set 
		Let \(f\) and \(g\) be two radial local strong solutions on \([0,T]\) with the
		same initial datum and with
		\[
		\sup_{0\leq t\leq T}
		\bigl(
		\|f(t)\|_{L^\infty}
		+
		\|g(t)\|_{L^\infty}
		\bigr)
		<\infty .
		\]
		Then \(f=g\) on \([0,T]\).
	\end{lemma}
	
	\begin{proof}
		Using the surface estimates of Proposition~\ref{lem-Sp}, the collision operator
		is locally Lipschitz in \(L^\infty_{\mathrm{rad}}\). Hence there exists a
		constant \(C_T>0\) such that
		\[
		\|Q_{3w}[f(t)]-Q_{3w}[g(t)]\|_{L^\infty}
		\leq
		C_T\|f(t)-g(t)\|_{L^\infty}.
		\]
		Therefore, we find
		\[
		\|f(t)-g(t)\|_{L^\infty}
		\leq
		C_T\int_0^t
		\|f(s)-g(s)\|_{L^\infty}\,ds .
		\]
		Gronwall's inequality gives
		\(
		\|f(t)-g(t)\|_{L^\infty}=0
		\)
		for every \(0\leq t\leq T\). Hence \(f=g\).
	\end{proof}
	
	\begin{lemma}
		\label{Lemma:positivity-preservation}
		Let
		\[
		X:=L^\infty_{\rm rad}(\mathbb R^d),
		\qquad
		X_+:=\{h\in X:h\geq0\ \text{a.e.}\}.
		\]
		Assume that
		\(
		h\in X_+.
		\)
		Let \(u\in C^1([0,T);X)\) be the unique strong radial solution of
		\[
		\partial_t u=Q_{3w}[u],
		\qquad
		u(0)=h.
		\]
		Then
		\(
		u(t,\cdot)\in X_+
		\)
		for every \(0\leq t<T\).
	\end{lemma}
	
	\begin{proof}
		We write the equation in gain-loss form. For a radial function \(f\), we define
		\[
		\begin{aligned}
			\mathcal G[f](k)
			&:=
			\int_{\mathscr S_k}
			\frac{
				V_{k,k_1,k-k_1}
				f(k_1)f(k-k_1)
			}{
				|\nabla\mathscr H_0^k(k_1)|
			}
			\,d\sigma(k_1)
			+
			2\int_{\mathscr S'_k}
			\frac{
				V_{k+k_2,k,k_2}
				f(k+k_2)f(k_2)
			}{
				|\nabla\mathscr H_1^k(k_2)|
			}
			\,d\sigma(k_2),
		\end{aligned}
		\]
		and
		\[
		\begin{aligned}
			\mathcal L[f](k)
			&:=
			\int_{\mathscr S_k}
			\frac{
				V_{k,k_1,k-k_1}
				\bigl(f(k_1)+f(k-k_1)\bigr)
			}{
				|\nabla\mathscr H_0^k(k_1)|
			}
			\,d\sigma(k_1)
			+
			2\int_{\mathscr S'_k}
			\frac{
				V_{k+k_2,k,k_2}
				f(k_2)
			}{
				|\nabla\mathscr H_1^k(k_2)|
			}
			\,d\sigma(k_2),
		\end{aligned}
		\]
		and
		\[
		\mathcal R[f](k)
		:=
		2\int_{\mathscr S'_k}
		\frac{
			V_{k+k_2,k,k_2}
			f(k+k_2)
		}{
			|\nabla\mathscr H_1^k(k_2)|
		}
		\,d\sigma(k_2).
		\]
		Then, we have
		\[
		Q_{3w}[f](k)
		=
		\mathcal G[f](k)
		-
		\bigl(\mathcal L[f](k)-\mathcal R[f](k)\bigr)f(k).
		\]
		We now set
		\[
		\mathcal A[f]:=\mathcal L[f]-\mathcal R[f].
		\]
		Thus the equation becomes
		\[
		\partial_t f+\mathcal A[f]f=\mathcal G[f].
		\]
		
		We first prove a local positivity result. By the surface estimates of
		Proposition~\ref{lem-Sp}, there exist constants \(C_G,C_A>0\), depending only
		on \(d,\gamma,\mathscr A\), such that for all \(f,g\in X\),
		\[
		\|\mathcal G[f]\|_X
		\leq
		C_G\|f\|_X^2, \ \ \ \ \|\mathcal G[f]-\mathcal G[g]\|_X
		\leq
		C_G\bigl(\|f\|_X+\|g\|_X\bigr)\|f-g\|_X,
		\]
		and
		\[
		\|\mathcal A[f]\|_X
		\leq
		C_A\|f\|_X,
		\qquad
		\|\mathcal A[f]-\mathcal A[g]\|_X
		\leq
		C_A\|f-g\|_X.
		\]
		Moreover, if \(f\in X_+\), then
		\(
		\mathcal G[f]\geq0
		\text{a.e.}
		\)
		
		We fix \(R>0\). We claim that there exists \(\tau_R>0\) such that, whenever
		\(
		h\in X_+,
		\|h\|_X\leq R,
		\)
		the Cauchy problem has a unique solution on \([0,\tau_R]\), and that solution
		is nonnegative.
		
		We set
		\[
		\mathcal E_{T,R}
		:=
		\left\{
		v\in C([0,T];X):
		v(t)\in X_+\ \text{for every }t\in[0,T],
		\quad
		\|v\|_{C([0,T];X)}\leq 2R
		\right\}.
		\]
		For \(v\in\mathcal E_{T,R}\),  we define
		\[
		(\Phi v)(t,k)
		:=
		h(k)
		\exp\left(
		-\int_0^t\mathcal A[v(s)](k)\,ds
		\right)
		+
		\int_0^t
		\mathcal G[v(\tau)](k)
		\exp\left(
		-\int_\tau^t\mathcal A[v(s)](k)\,ds
		\right)
		\,d\tau .
		\]
		The exponential factors are strictly positive. Since \(h\geq0\) and
		\(\mathcal G[v]\geq0\) whenever \(v\geq0\), we have
		\(
		\Phi v(t)\in X_+
		\)
		for every \(t\in[0,T]\).
		
		We next show that \(\Phi\) maps \(\mathcal E_{T,R}\) into itself for \(T\)
		small. If \(v\in\mathcal E_{T,R}\), then we bound
		\[
		\|\mathcal A[v(t)]\|_X\leq 2C_AR,
		\qquad
		\|\mathcal G[v(t)]\|_X\leq 4C_GR^2.
		\]
		Therefore, we find
		\[
		\|\Phi v(t)\|_X
		\leq
		\exp(2C_ART)
		\bigl(R+4C_GR^2T\bigr).
		\]
		Choosing \(T>0\) sufficiently small, depending only on \(R\), we obtain
		\[
		\exp(2C_ART)
		\bigl(R+4C_GR^2T\bigr)
		\leq
		2R.
		\]
		Thus \(\Phi\mathcal E_{T,R}\subseteq\mathcal E_{T,R}\).
		We now prove that \(\Phi\) is a contraction for \(T\) still smaller. Let
		\(v,w\in\mathcal E_{T,R}\). Using
		\[
		\|\mathcal A[v]-\mathcal A[w]\|_X
		\leq
		C_A\|v-w\|_X,\ \ \ \ 
		\|\mathcal G[v]-\mathcal G[w]\|_X
		\leq
		4C_GR\|v-w\|_X,
		\]
		we obtain
		\[
		\begin{aligned}
			\|\Phi v-\Phi w\|_{C([0,T];X)}
			&\leq
			\exp(2C_ART)
			\Bigl(
			C_ART
			+
			4C_GRT
			+
			4C_AC_GR^2T^2
			\Bigr)
			\|v-w\|_{C([0,T];X)}.
		\end{aligned}
		\]
		Choosing \(T>0\) smaller if necessary, the coefficient on the right-hand side
		is strictly less than \(1\). Therefore \(\Phi\) is a contraction on
		\(\mathcal E_{T,R}\).
		
		By the Banach fixed-point theorem, \(\Phi\) has a unique fixed point
		\(u\in\mathcal E_{T,R}\). Since \(u=\Phi u\), differentiating the
		variation-of-constants formula gives
		\[
		\partial_t u+\mathcal A[u]u=\mathcal G[u],
		\]
		or equivalently
		\[
		\partial_t u=Q_{3w}[u],
		\qquad
		u(0)=h.
		\]
		Moreover, because \(u\in\mathcal E_{T,R}\), we have
		\(
		u(t)\in X_+
		\)
		for every \(t\in[0,T]\). This proves local positivity.
		
		Now let \(u\in C^1([0,T);X)\) be the strong radial solution with
		nonnegative initial datum \(h\). We define
		\[
		T_*
		:=
		\sup
		\left\{
		s\in[0,T):
		u(t,\cdot)\in X_+
		\ \text{for every }0\leq t\leq s
		\right\}.
		\]
		It is clear that \(T_*>0\).
		
		We claim that \(T_*=T\). Suppose, for contradiction, that \(T_*<T\). Since
		\(u\in C([0,T);X)\), and since \(X_+\) is a closed cone in \(X\), we have
		\(
		u(T_*,\cdot)\in X_+.
		\)
		We set
		\(
		h_*:=u(T_*,\cdot).
		\)
		Then \(h_*\in X_+\) and \(h_*\in X\). Applying the local positivity result with
		initial datum \(h_*\), we obtain a nonnegative strong radial solution
		\(
		w\in C^1([0,\tau];X)
		\)
		of
		\[
		\partial_t w=Q_{3w}[w],
		\qquad
		w(0)=h_*,
		\]
		for some \(\tau>0\).
		
		On the other hand, the shifted function
		\(
		\widetilde u(s,\cdot):=u(T_*+s,\cdot)
		\)
		also solves the same Cauchy problem on the interval where it is defined. By
		local uniqueness in \(X=L^\infty_{\rm rad}(\mathbb R^d)\),
		\(
		\widetilde u(s,\cdot)=w(s,\cdot)
		\)
		for \(0\leq s<\min\{\tau,T-T_*\}\). Hence
		\(
		u(T_*+s,\cdot)\in X_+
		\)
		for all such \(s\). This contradicts the definition of \(T_*\). Therefore
		\(T_*=T\), and the solution remains nonnegative on the whole interval of
		existence.
	\end{proof}
	
	\begin{lemma}
		\label{lemma:local-analytic-banach}
		Let \(h\in X=L^\infty_{\rm rad}(\mathbb R^d)\), and set
		\(
		H:=\|h\|_X.
		\)
		If \(H>0\), then there
		exists a unique strong radial solution
		\(
		u\in C^1((-T_h,T_h);X)
		\)
		of
		\[
		\partial_t u=Q_{3w}[u],
		\qquad
		u(0)=h,
		\]
		for some
		\(
		T_h\geq \frac{1}{4C_BH}.
		\)
		Moreover, \(u\) is real analytic as an \(X\)-valued function of time.
		
		More precisely, \(u\) admits the power-series representation
		\[
		u(t)=\sum_{n=0}^\infty t^n A_n
		\]
		in \(X\) for \(|t|<1/(C_BH)\), where
		\(
		A_0=h
		\)
		and
		\[
		A_{n+1}
		=
		\frac{1}{n+1}
		\sum_{\substack{i,j\geq0\\ i+j=n}}
		\mathcal B(A_i,A_j),
		\qquad n\geq0.
		\]
		The coefficients satisfy
		\[
		A_n\in X,
		\qquad
		\|A_n\|_X\leq (C_BH)^nH.
		\]
		
		If \(H=0\), then the unique solution is \(u\equiv0\). 
	\end{lemma}
	
	\begin{proof}
		We first construct the formal coefficients. Since \(A_0=h\in X\), and since
		\(\mathcal B:X\times X\to X\) is bilinear and bounded by
		Lemma~\ref{lemma:bilinear-bound}, the recursive formula implies inductively
		that
		\(
		A_n\in X
		\)
		for every \(n\geq0\). In particular, every coefficient is radial.
		
		We prove the coefficient estimate by induction. The estimate is immediate for
		\(n=0\). We assume that
		\[
		\|A_\ell\|_X\leq (C_BH)^\ell H
		\]
		for \(0\leq \ell\leq n\). Then we have
		\[
		\begin{aligned}
			\|A_{n+1}\|_X
			&\leq
			\frac{1}{n+1}
			\sum_{\substack{i,j\geq0\\ i+j=n}}
			\|\mathcal B(A_i,A_j)\|_X                                      
			\leq
			\frac{C_B}{n+1}
			\sum_{\substack{i,j\geq0\\ i+j=n}}
			\|A_i\|_X\|A_j\|_X                                               \\
			&\leq
			\frac{C_B}{n+1}
			\sum_{\substack{i,j\geq0\\ i+j=n}}
			(C_BH)^iH(C_BH)^jH                                               
			=
			C_B(C_BH)^nH^2
			=
			(C_BH)^{n+1}H.
		\end{aligned}
		\]
		This closes the induction.
		
		Therefore the series
		\(
		\sum_{n=0}^\infty t^nA_n
		\)
		converges absolutely in \(X\) whenever
		\(
		|t|<\frac{1}{C_BH}.
		\)
		Furthermore, for every \(0<\rho<1/(C_BH)\),
		\[
		\sum_{n=1}^\infty n\rho^{n-1}\|A_n\|_X<\infty.
		\]
		Hence the series may be differentiated term by term for \(|t|<1/(C_BH)\), and
		\[
		\frac{d}{dt}\sum_{n=0}^\infty t^nA_n
		=
		\sum_{n=0}^\infty (n+1)t^nA_{n+1}.
		\]
		
		It remains to check that the collision operator may be applied termwise. We let
		\[
		u(t):=\sum_{n=0}^\infty t^nA_n.
		\]
		Since \(\mathcal B\) is a bounded bilinear map on \(X\), absolute convergence
		of the series gives
		\[
		\mathcal B(u(t),u(t))
		=
		\sum_{n=0}^\infty
		t^n
		\sum_{\substack{i,j\geq0\\ i+j=n}}
		\mathcal B(A_i,A_j)
		\]
		in \(X\). By the defining recursion for \(A_{n+1}\), this becomes
		\[
		\mathcal B(u(t),u(t))
		=
		\sum_{n=0}^\infty
		(n+1)t^nA_{n+1}
		=
		\partial_t u(t).
		\]
		Since \(Q_{3w}[u]=\mathcal B(u,u)\), we obtain
		\(
		\partial_t u=Q_{3w}[u]
		\)
		in \(X\), and \(u(0)=h\).
		
		Uniqueness follows from the Lipschitz estimate in
		Lemma~\ref{lemma:bilinear-bound}. Indeed, if \(u\) and \(v\) are two solutions
		on a common interval and are bounded in \(X\), then
		\[
		\|Q_{3w}[u(t)]-Q_{3w}[v(t)]\|_X
		\leq
		C_B\bigl(\|u(t)\|_X+\|v(t)\|_X\bigr)\|u(t)-v(t)\|_X.
		\]
		Gronwall's inequality gives \(u=v\).
		The power-series representation proves real analyticity in time as an
		\(X\)-valued function.
	\end{proof}
	\section{Proof of Theorem \ref{maintheorem}}
	
	Let \(T_{\max}\in(0,+\infty]\) be the maximal time of existence of the
	nonnegative strong radial solution constructed by the local theory. We will prove
	that
	\[
	T_{\max}=+\infty.
	\]
	
	We now set
	\(
	N(t):=\|f(t,\cdot)\|_{L^\infty(\mathbb R^d)}.
	\)
	By the local positivity result, the solution remains nonnegative on its
	maximal interval of existence:
	\[
	f(t,k)\geq0
	\]
	for a.e. \(k\in\mathbb R^d\) and every \(0\leq t<T_{\max}\).
	
	By the standard energy
	conservation (see \cite{nguyen2017quantum,ToanBinh}), we have
	\[
	\int_0^\infty f(t,r)\,\omega(r)\,r^{d-1}\,\mathrm dr
	=
	\int_0^\infty f_0(r)\,\omega(r)\,r^{d-1}\,\mathrm dr
	=
	\mathcal M_1
	\]
	for every \(0\leq t<T_{\max}\).
	
	We now derive an a priori estimate for \(N(t)\). We now recall
	\[
	\partial_t f(t,k)
	+
	\bigl(\mathcal L[f](t,k)-\mathcal R[f](t,k)\bigr)f(t,k)
	=
	\mathcal G[f](t,k).
	\]
	Since
	\(
	f\geq0
	\text{ and }
	\mathcal L[f]\geq0,
	\)
	we have, for a.e. \(k\),
	\[
	\partial_t f(t,k)
	\leq
	\mathcal G[f](t,k)+\mathcal R[f](t,k)f(t,k).
	\]
	Therefore, we can bound
	\[
	f(t,k)
	\leq
	f_0(k)
	+
	\int_0^t
	\Big(
	\mathcal G[f](\tau,k)
	+
	\mathcal R[f](\tau,k)f(\tau,k)
	\Big)
	\,\mathrm d\tau .
	\]
	
	We estimate the integrand. Since \(f\geq0\), we have
	\[
	\begin{aligned}
		&\mathcal G[f](t,k)+\mathcal R[f](t,k)f(t,k)\ 
		=
		\int_{\mathscr S_k}
		\mathrm d\sigma(k_1)\,
		\frac{
			V_{k,k_1,k-k_1}
			f(t,k_1)f(t,k-k_1)}
		{|\nabla\mathscr H_0^k(k_1)|}
		\\
		&\quad
		+
		2\int_{\mathscr S'_k}
		\mathrm d\sigma(k_2)\,
		\frac{
			V_{k+k_2,k,k_2}
			f(t,k+k_2)f(t,k_2)}
		{|\nabla\mathscr H_1^k(k_2)|}
		+
		2f(t,k)
		\int_{\mathscr S'_k}
		\mathrm d\sigma(k_2)\,
		\frac{
			V_{k+k_2,k,k_2}
			f(t,k+k_2)}
		{|\nabla\mathscr H_1^k(k_2)|}.
	\end{aligned}
	\]
	Using
	\(
	0\leq f(t,\cdot)\leq N(t),
	\)
	we obtain
	\[
	\begin{aligned}
		&\mathcal G[f](t,k)+\mathcal R[f](t,k)f(t,k)
		\ 
		\leq
		N(t)
		\int_{\mathscr S_k}
		\mathrm d\sigma(k_1)\,
		\frac{
			V_{k,k_1,k-k_1}
			f(t,k_1)}
		{|\nabla\mathscr H_0^k(k_1)|}
		\\
		&\quad
		+
		2N(t)
		\int_{\mathscr S'_k}
		\mathrm d\sigma(k_2)\,
		\frac{
			V_{k+k_2,k,k_2}
			f(t,k_2)}
		{|\nabla\mathscr H_1^k(k_2)|}
		+
		2N(t)
		\int_{\mathscr S'_k}
		\mathrm d\sigma(k_2)\,
		\frac{
			V_{k+k_2,k,k_2}
			f(t,k+k_2)}
		{|\nabla\mathscr H_1^k(k_2)|}.
	\end{aligned}
	\]
	
	The first integral is controlled by Proposition~\ref{lem-Sp}
	\[
	\int_{\mathscr S_k}
	\mathrm d\sigma(k_1)\,
	\frac{
		V_{k,k_1,k-k_1}
		f(t,k_1)}
	{|\nabla\mathscr H_0^k(k_1)|}
	\leq
	\mathcal C_S\mathcal M_1.
	\]
	Similarly, we can also bound
	\[
	\int_{\mathscr S'_k}
	\mathrm d\sigma(k_2)\,
	\frac{
		V_{k+k_2,k,k_2}
		f(t,k_2)}
	{|\nabla\mathscr H_1^k(k_2)|}
	\leq
	\mathcal C_{S'}\mathcal M_1.
	\]
	
	It remains to estimate the last term. Using the parametrization of
	\(\mathscr S'_k\), we write
	\[
	u:=|k_2|,\ \ \ 
	r_+(u):=|k+k_2|.
	\]
	On the resonant surface, we have
	\(
	\omega(r_+(u))=\omega(|k|)+\omega(u).
	\)
	Thus, for \(\omega(r)=r^\gamma\),
	\(
	r_+(u)=\bigl(|k|^\gamma+u^\gamma\bigr)^{1/\gamma}.
	\) The surface estimate gives
	\[
	\begin{aligned}
		\int_{\mathscr S'_k}
		\mathrm d\sigma(k_2)\,
		\frac{
			V_{k+k_2,k,k_2}
			f(t,k+k_2)}
		{|\nabla\mathscr H_1^k(k_2)|}
		\
		&\leq \
		\mathcal C_{S'}
		\int_0^\infty
		f(t,r_+(u))\,u^{\gamma+d-1}\,\mathrm du .
	\end{aligned}
	\]
	We change variables
	\(
	r=r_+(u)=\bigl(|k|^\gamma+u^\gamma\bigr)^{1/\gamma}.
	\)
	Then, we find
	\[
	u=(r^\gamma-|k|^\gamma)^{1/\gamma},\ \ \ \mathrm du
	=
	r^{\gamma-1}
	(r^\gamma-|k|^\gamma)^{1/\gamma-1}
	\,\mathrm dr.
	\]
	Hence, we obtain
	\[
	u^{\gamma+d-1}\,\mathrm du
	=
	u^d r^{\gamma-1}\,\mathrm dr
	\leq
	r^{\gamma+d-1}\,\mathrm dr.
	\]
	Therefore, we have
	\[
	\begin{aligned}
		\int_0^\infty
		f(t,r_+(u))\,u^{\gamma+d-1}\,\mathrm du
		&\leq
		\int_{|k|}^\infty
		f(t,r)\,r^{\gamma+d-1}\,\mathrm dr\leq
		\mathcal M_1.
	\end{aligned}
	\]
	Thus the last term is bounded by
	\(
	\mathcal C_{S'}\mathcal M_1.
	\)
	
	Combining the three estimates, we obtain
	\[
	\mathcal G[f](t,k)+\mathcal R[f](t,k)f(t,k)
	\leq
	C_0\mathcal M_1N(t),
	\]
	where
	\(
	C_0:=\mathcal C_S+4\mathcal C_{S'}.
	\)
	Consequently, we have
	\[
	f(t,k)
	\leq
	f_0(k)
	+
	C_0\mathcal M_1
	\int_0^t N(\tau)\,\mathrm d\tau.
	\]
	Taking the essential supremum over \(k\), we get
	\[
	N(t)
	\leq
	N(0)
	+
	C_0\mathcal M_1
	\int_0^t N(\tau)\,\mathrm d\tau.
	\]
	By Gronwall's inequality, we have
	\[
	N(t)
	\leq
	N(0)\exp(C_0\mathcal M_1t)
	\]
	for every
	\(
	0\leq t<T_{\max}.
	\)
	
	If \(T_{\max}<+\infty\), then the previous estimate gives
	\[
	\limsup_{t\uparrow T_{\max}}N(t)
	\leq
	N(0)\exp(C_0\mathcal M_1T_{\max})
	<+\infty.
	\]
	This contradicts the blow-up alternative. Therefore, we deduce
	\(
	T_{\max}=+\infty.
	\)
	Hence the nonnegative strong radial solution is global in time.
	
	Finally, we prove the asserted time analyticity. We fix \(t_0\geq0\), and set
	\(
	h:=f(t_0,\cdot)\in L^\infty_{\rm rad}(\mathbb R^d).
	\)
	By Lemma~\ref{lemma:local-analytic-banach}, the Cauchy problem with initial
	datum \(h\) has a unique \(L^\infty_{\rm rad}\)-valued analytic solution
	near \(s=0\). On the other hand, the shifted function
	\(
	\widetilde f(s,\cdot):=f(t_0+s,\cdot)
	\)
	is also a strong radial solution with the same initial datum \(h\), for all
	\(s\) sufficiently small such that \(t_0+s\geq0\). Therefore,
	\[
	\widetilde f(s,\cdot)
	=
	\sum_{n=0}^\infty s^nA_n^{(t_0)}
	\]
	in \(L^\infty_{\rm rad}(\mathbb R^d)\) for \(s\) sufficiently small. If
	\(t_0>0\), this gives a two-sided analytic expansion around \(t_0\). If
	\(t_0=0\), it gives right-analyticity at the initial time. This completes the proof of Theorem~\ref{maintheorem}.

	\bibliographystyle{plain}
	\bibliography{WaveTurbulence}
\end{document}